\documentclass[11pt]{article}

\usepackage[T1]{fontenc}
\usepackage[sc]{mathpazo}
\usepackage{amsmath}
\usepackage{amssymb}
\usepackage{enumerate}
\usepackage{amsthm}
\usepackage{amsfonts,mathrsfs}
\usepackage[style]{fncychap}
\usepackage{graphicx} 
\usepackage{geometry} 
\usepackage{eepic}
\usepackage{ifthen}
\newboolean{ElectronicVersion}
\setboolean{ElectronicVersion}{true} 

\usepackage{hyperref}
\hypersetup{pdfpagemode=UseNone}

\renewcommand{\t}{{\scriptscriptstyle\mathsf{T}}}

\newcommand{\setft}[1]{\mathrm{#1}}
\newcommand{\lin}[1]{\setft{L}\left(#1\right)}
\newcommand{\density}[1]{\setft{D}\left(#1\right)}

\newcommand{\trans}[1]{\setft{T}\left(#1\right)}

\def\complex{\mathbb{C}}
\def\real{\mathbb{R}}
\def\natural{\mathbb{N}}

\def\I{\mathbb{1}}

\newenvironment{mylist}[1]{\begin{list}{}{
    \setlength{\leftmargin}{#1}
    \setlength{\rightmargin}{0mm}
    \setlength{\labelsep}{2mm}
    \setlength{\labelwidth}{8mm}
    \setlength{\itemsep}{0mm}}}
    {\end{list}}

\def\ot{\otimes}
\def\CP{\mathrm{CP}}

\newcommand{\out}[2]{| #1\rangle\langle #2 |}

\newcommand{\Innerm}[3]{\left\langle #1 \left| #2 \right| #3 \right\rangle}
\newcommand{\defeq}{\stackrel{\smash{\textnormal{\tiny def}}}{=}}

\newcommand{\pa}[1]{(#1)}
\newcommand{\Pa}[1]{\left(#1\right)}

\newcommand{\set}[1]{\{#1\}}

\def\Jamiolkowski{J}
\newcommand{\jam}[1]{\Jamiolkowski\pa{#1}}

\DeclareMathOperator{\vectorize}{vec}
\newcommand{\col}[1]{\vectorize\pa{#1}}
\newcommand{\row}[1]{\vectorize\pa{#1}^{\dagger}}

\DeclareMathOperator{\trace}{Tr}

\newcommand{\Ptr}[2]{\trace_{#1}\Pa{#2}}

\newcommand{\Tr}[1]{\Ptr{}{#1}}

\def\cH{\mathcal{H}}
\def\cK{\mathcal{K}}

\def\bB{\mathbf{B}}

\def\bS{\mathbf{S}}

\def\rH{\mathrm{H}}
\def\rM{\mathrm{M}}
\def\rS{\mathrm{S}}

\def\sM{\mathscr{M}}

\newtheorem{thrm}{Theorem}[section]
\newtheorem{lem}[thrm]{Lemma}
\newtheorem{prop}[thrm]{Proposition}

\theoremstyle{definition}

\newtheorem{remark}[thrm]{Remark}

\numberwithin{equation}{section}

\newcounter{questionnumber}

\begin{document}

\title{Entropy, stochastic matrices, and quantum operations}

\author{Lin Zhang\footnote{E-mail: godyalin@163.com; linyz@hdu.edu.cn}\\[1mm]
  {\it\small Institute of Mathematics, Hangzhou Dianzi University, Hangzhou 310018,
  PR~China}}

\date{}
\maketitle

\begin{abstract}

The goal of the present paper is to derive some conditions on
saturation of (strong) subadditivity inequality for the stochastic
matrices. The notion of relative entropy of stochastic matrices is
introduced by mimicking quantum relative entropy. Some properties of
this concept are listed and the connection between the entropy of
the stochastic quantum operations and that of stochastic matrices
are discussed.

\end{abstract}

\section{Introduction}

If the column vectors
$$
\mathbf{p} = [p_{1},\ldots,p_{N}]^{\t} \in \real^{N},\quad\mathbf{q}
= [q_{1},\ldots,q_{N}]^{\t} \in \real^{N}
$$
are two probability
distributions, the \emph{Shannon entropy} of $\mathbf{p}$ is defined
by
$$
\rH(\mathbf{p}) \defeq - \sum_{i=1}^{N} p_{i}\log_{2}p_{i},
$$
where $x\log_{2}x$ is set to $0$ if $x = 0$, and the \emph{relative
entropy} of $\mathbf{p}$ and $\mathbf{q}$ is defined by
$$
\rH(\mathbf{p}||\mathbf{q}) \defeq \sum_{i=1}^N
p_i\log_2\frac{p_i}{q_i}
$$
when $\mathbf{p}$ is absolute continuously with respect to
$\mathbf{q}$; $\rH(\mathbf{p}||\mathbf{q}) = +\infty$ otherwise.

Let $B = [b_{ij}]$ be a $N \times N$ \emph{bi-stochastic matrix},
that is, $b_{ij} \geqslant 0$, and
$$
\sum_{i=1}^{N} b_{ij} =
\sum_{j=1}^{N} b_{ij} = 1
$$
for each $i, j=1, \ldots, N$. Let $\pi$
be a permutation of the set $\set{1,\ldots,N}$. For any $i, j \in
\set{1,\ldots,N}$, we define $c_{ij} = 1$ when $i = \pi(j)$ and
$c_{ij} = 0$ when $i \neq \pi(j)$. Then the matrix $C = [c_{ij}]$ is
called a permutation matrix. Let $\bS_{N}$ be the set of all $N
\times N$ permutation matrices and $\bB_{N}$ be the convex hull
$\bB_{N}$ of $\bS_{N}$. The well-known Birkhoff-von Neumann theorem
indicates that $\bB_{N}$ is the set of all $N \times N$
bi-stochastic matrix.

We only consider finite dimensional complex Hilbert spaces. A state
$\rho$ of a quantum system, described by Hilbert space $\cH$, is a
positive semi-definite matrix of trace one and call it the
\emph{density matrix}. The set of all density matrices of $\cH$ is
denoted by $\density{\cH}$, if $\rho\in \density{\cH}$ is
invertible, then $\rho$ is said to be \emph{faithful}. If  $\rho$
and $\sigma$ are two quantum states, then the \emph{von Neumann
entropy} of $\rho$ is defined by
$$
\rS\pa{\rho} \defeq - \Tr{\rho\log_{2}\rho},
$$
the quantum \emph{relative entropy} between $\rho$ and $\sigma$ is
defined by
$$
\rS\pa{\rho||\sigma} \defeq \Tr{\rho(\log_{2}\rho -
\log_{2}\sigma)}
$$
if $\mathbf{supp}(\rho)\subseteq
\mathbf{supp}(\sigma)$; $\rS\pa{\rho||\sigma} = +\infty$ otherwise,
see \cite{Nielsen}.

Let $\cH$ and $\cK$ be two Hilbert spaces, $\lin{\cH, \cK}$ be the
set of all linear operators from $\cH$ to $\cK$, denote $\lin{\cH,
\cH}$ by $\lin{\cH}$. Let $\trans{\cH,\cK}$ denote the set of all
\emph{linear super-operators} from $\lin{\cH}$ to $\lin{\cK}$,
similarly, denote $\trans{\cH,\cH}$ by $\trans{\cH}$. We say that
$\Phi \in \trans{\cH,\cK}$ is \emph{completely positive} (CP) if for
each $k \in \natural$,
$$
\Phi \ot \I_{\rM_{k}(\complex)}: \lin{\cH}
\ot \rM_{k}(\complex) \to \lin{\cK}\ot \rM_{k}\pa{\complex}
$$
is
positive, where $\rM_{k}(\complex)$ is the set of all $k\times k$
complex matrices. It follows from the famous theorems of Choi
\cite{Choi} and Kraus \cite{Kraus} that $\Phi$ can be represented in
the form
$$
\Phi = \sum_{j}\mathbf{Ad}_{M_{j}},
$$
where $\set{M_j}_{j=1}^n
\subseteq \lin{\cH, \cK}$, that is,
$$\Phi(X) =
\sum_{j=1}^nM_jXM_j^\dagger,\, X \in \lin{\cH}.
$$
Throughout the
present paper, $\dagger$ means adjoint operation of some operator.
Denote by $\CP(\cH,\cK)(\CP(\cH))$ the set of all linear CP
super-operators in $\trans{\cH,\cK}(\trans{\cH})$.

The so-called \emph{quantum operation} is just a trace
non-increasing $\Phi \in \CP(\cH,\cK)$, if $\Phi$ is
trace-preserving, then it is called \emph{stochastic}; if $\Phi$ is
stochastic and unit-preserving, then it is called
\emph{bi-stochastic}.

The famous \emph{Jamio{\l}kowski isomorphism}
$$
J :
\trans{\cH}\longrightarrow\lin{\cH\ot\cH}
$$
transforms each
$\Phi\in\trans{\cH}$ into an operator $\jam{\Phi} \in
\lin{\cH\ot\cH}$, where
$$
\jam{\Phi} = \Phi\ot\I_{\lin{\cH}}(\col{\I_{\cH}}\row{\I_{\cH}}).
$$
If $\Phi \in \CP(\cH)$, then $\jam{\Phi}$ is a positive
semi-definite operator, in particular, if $\Phi$ is stochastic, then
$\frac1N\jam{\Phi}$ is a state on $\cH\ot\cH$. If $\Phi \in
\CP(\cH)$ is a stochastic quantum operation, we denote the von
Neumann entropy $\rS(\frac1N\jam{\Phi})$ of $\frac1N\jam{\Phi}$ by
$\rS^{\mathbf{map}}(\Phi)$ and call it the \emph{map entropy}
\cite{Roga}, which describes the decoherence induced by the quantum
operation $\Phi$.

\section{On saturation of classical relative entropy}

In order to obtain the condition for saturation of classical
relative entropy, we need the following lemmas.

\begin{lem}\label{Lind}(\cite{Lindblad})
Let $\cH$ be a Hilbert space, $\rho$ and $\sigma$ be two states of
$\cH$. If $\Phi \in \CP(\cH)$ is stochastic, then
$\rS(\Phi(\rho)||\Phi(\sigma)) \leqslant \rS(\rho||\sigma)$.
\end{lem}

\begin{lem}\label{CKL}(\cite{Li})
Let $\{A_{1}, \ldots, A_{k}\} \subseteq \lin{\mathbb{C}^{n}}$ and
$\{B_{1}, \ldots, B_{k}\} \subseteq \lin{\mathbb{C}^{m}}$ be two
commuting families of Hermitian matrices. Then there exist unitary
matrices $U \in \lin{\complex^n}$ and $V \in \lin{\complex^m}$ such
that $U^\dagger A_{j} U$ and $V B_{j} V^\dagger$ are diagonal
matrices with diagonals $\mathbf{a}_{j} = [a_{1j}, \ldots,
a_{nj}]^\t$ and $\mathbf{b}_{j} = [b_{1j},\ldots,b_{mj}]^\t$,
respectively, for $j=1,\ldots,k$. Then the following conditions are
equivalent:
\begin{enumerate}[(i)]
\item\label{claim:i} There is a super-operator $\Phi \in \CP(\complex^n, \complex^m)$ such that $\Phi(A_j) = B_j(j=1, \ldots, k)$.
\item\label{claim:ii} There is an $m \times n$ non-negative matrix $D = [d_{\mu\nu}]$ such that $[b_{ij}] = D [a_{ij}]$.
\end{enumerate}
Moreover, if the statement \eqref{claim:ii} is satisfied, then
$\Phi$ is bi-stochastic if and only if $D$ is bi-stochastic.
\end{lem}

\begin{thrm}\label{th:Sentropy:equality}
Let $T$ be a $N\times N$ stochastic matrix, $\mathbf{p} =
[p_1,p_2,\ldots,p_N]^\t$ and $\mathbf{q} = [q_1,q_2,\ldots,q_N]^\t$
be two $N$-dimensional probability distributions. Then $\rH(T
\mathbf{p}||T\mathbf{q})\leqslant\rH(\mathbf{p}||\mathbf{q})$.
Moreover, for each $1\leqslant k \leqslant N$, $p_k, q_k>0$, then
$\rH(T \mathbf{p}||T \mathbf{q})=\rH(\mathbf{p}||\mathbf{q})$ if and
only if the following conditions hold:
\begin{enumerate}[(i)]
\item $\mathbf{p} = \bigoplus_{k=1}^{K} \mu_k \mathbf{p}_k \ot \mathbf{r}_k$ and $\mathbf{q} = \bigoplus_{k=1}^{K} \nu_{k} \mathbf{q}_k \ot
\mathbf{r}_k$, where $\mathbf{p}_k, \mathbf{q}_k$ denote the
$m_{k}$-dimensional probability vectors, and $\mathbf{r}_k$ denotes
the $n_k$-dimensional probability vectors, and $\mu_{k},\nu_{k}
\geqslant 0, k = 1, \ldots, K; \sum_{k=1}^K \mu_{k} = \sum_{k=1}^K
\nu_{k}=1$, $\sum_{k=1}^K m_k n_k=N$;
\item $T = \bigoplus_{k=1}^{K}\pi_k \ot T_{k}$, $\pi_{k} \in \mathbf{S}_{m_k}$ and $T_k$ is $n_k\times n_k$ stochastic matrix for each $k=1, \ldots, K$.
\end{enumerate}
\end{thrm}

\begin{proof}
Let $\rho$, $\sigma$, $\rho'$ and $\sigma'$ be diagonal matrices
with diagonal $\mathbf{p}$, $\mathbf{q}$, $T\mathbf{p}$ and
$T\mathbf{q}$, respectively. Then it follows from the Lemma
\ref{CKL} that there is a stochastic $\Phi\in\CP(\complex^n,
\complex^m)$ such that $\Phi(\rho)=\rho'$,  $\Phi(\sigma) =
\sigma'$. Note that $\rH(\mathbf{p}||\mathbf{q}) =
\rS(\rho||\sigma)$ and $\rH(T\mathbf{p}||T\mathbf{q}) =
\rS(\rho'||\sigma')$, so by Lemma 1 we have
$\rH(T\mathbf{p}||T\mathbf{q}) \leqslant
\rH(\mathbf{p}||\mathbf{q})$. Moreover, if for each $1\leqslant k
\leqslant N$, $p_k, q_k>0$, then the states $\rho$, $\sigma$,
$\Phi(\rho) = \rho'$ and $\Phi(\sigma)=\sigma'$ are faithful. If
$\rH(T\mathbf{p}||T\mathbf{q}) = \rH(\mathbf{p}||\mathbf{q})$, then
$\rS(\rho||\sigma) = \rS(\Phi(\rho)||\Phi(\sigma))$. By a result in
\cite{Mosonyi}
$$
\rS(\Phi(\rho)||\Phi(\sigma)) = \rS(\rho||\sigma)
$$
if and only if the following statements hold:
\begin{enumerate}[(1)]
\item $\cH$ and $\cK$ can be decomposed by the form
$\cH = \bigoplus_{k=1}^{K} \cH^L_k \ot \cH^R_k$,
$\cK = \bigoplus_{k=1}^{K} \cK^L_k \ot \cK^R_k$,
where $\dim\cH^L_k = \dim\cK^L_k$.
\item If $\Phi_k$ is the restriction of $\Phi$ to $\lin{\cH^L_k\ot\cH^R_k}$,
then $\Phi_k \in \trans{\cH^L_k \ot \cH^R_k, \cK^L_k \ot \cK^R_k}$
and it can be factorized into the form $\Phi_k = \mathbf{Ad}_{U_k}
\ot \Phi^R_k$, where $U_k: \cH^L_k \longrightarrow \cK^L_k$ is
unitary operator and $\Phi^R_k \in \trans{\cH^R_k, \cK^R_k}$ is
stochastic, $k=1, \ldots, K$.
\item The state $\rho$ decomposes as
$\rho=\bigoplus_{k=1}^{K}p_k \rho^L_k \ot \omega^R_k$,
$\sigma=\bigoplus_{k=1}^{K}q_k \sigma^L_k\ot\omega^R_k$,
where all the operators are density operators, and
$\set{p_k}_{k=1}^K$ and $\set{q_k}_{k=1}^K$ are probability distributions.
\end{enumerate}
Therefore, it follows that the result can be proved by the above
decomposition of $\Phi$.
\end{proof}

\begin{remark}
In \cite{Zhang}, it was shown that $\rS(\Phi(\rho)) = \rS(\rho)$ if
and only if $\Phi^\dagger\circ\Phi(\rho) = \rho$ while the explicit
construction of the state $\rho$ and the quantum operation $\Phi$
are given. We can employ the mentioned result to give an explicit
construction for $T,\mathbf{p}$ in the identity: $\rH(T\mathbf{p}) =
\rH(\mathbf{p})$. The proof is trivial and is omitted.
\end{remark}

\section{Relative entropy of stochastic matrices}

In this section, the entropy of stochastic matrices is discussed.
For the entropy of stochastic matrices, more details can be found in
\cite{Slomczynski}. We will go deeper within the entropy concerning
stochastic matrices and derive some conditions on the (strong)
additivity for the stochastic matrices. The notion of relative
entropy of stochastic matrices is introduced by mimicking quantum
relative entropy. Some properties of this concept are listed and the
connection between the entropy of the stochastic quantum operations
and that of stochastic matrices are discussed.

To be specific, for any $N\times N$ stochastic matrix $T =
[t_{\mu\nu}]$, the \emph{weighted entropy} \cite{Slomczynski} of $T$
by a probability vector $\mathbf{p} = [p_{1},\ldots,p_{N}]^{\t}$ is
defined by $\rH_{\mathbf{p}}(T) = \sum_{\nu=1}^{N} p_{\nu}
\rH(\mathbf{t}_{\nu})$, where $T = [\mathbf{t}_1, \ldots,
\mathbf{t}_{N}]$ and $\mathbf{t}_{\nu} = [t_{1\nu}, \ldots,
t_{N\nu}]^{\t}$ is the $\nu$th column vector of $T$. In particular,
$\rH(T) = \frac{1}{N} \sum_{\nu=1}^{N} \rH(\mathbf{t}_{\nu})$ is
defined for $\mathbf{p} = \frac{1}{N}[1, \ldots, 1]^{\t}$.

For any two $N\times N$ stochastic matrices $A$ and $B$, the
\emph{relative entropy} between $A$ and $B$ with respect to a
probability vector $\mathbf{p} = [p_{1}, \ldots, p_{N}]^{\t}$ is
defined by $\rH_{\mathbf{p}}(A||B) = \sum_{\nu=1}^{N}p_{\nu}
\rH(\mathbf{a}_{\nu}||\mathbf{b}_{\nu})$, where $\mathbf{a}_{\nu}$
and $\mathbf{b}_{\nu}$ are the $\nu$th column vectors of $A$ and
$B$, respectively. Similarly, $\rH(A||B) = \frac{1}{N}
\sum_{\nu=1}^{N} \rH(\mathbf{a}_{\nu}||\mathbf{b}_{\nu})$ is defined
for $\mathbf{p} = \frac{1}{N}[1, \ldots, 1]^{\t}$.

The following conclusions are immediate. That is,
$\rH_{\mathbf{p}}(\cdot)$ is a nonnegative and concave function;
$\rH_{\mathbf{p}}(\cdot||\cdot)$ is a jointly convex function.

In what follows, the monotonicity of relative entropy of stochastic
matrices is obtained.


\begin{thrm}\label{Mon:stochasticmatrices}
If $T,A,B$ are all $N\times N$ stochastic matrices, then
$$\rH_{\mathbf{p}}(TA||TB) \leqslant \rH_{\mathbf{p}}(A||B),$$
where $\mathbf{p}$ is an $N$-dimensional probability vector.
Moreover, if all the components of $\mathbf{p}$ are positive, then
$$\rH_{\mathbf{p}}(TA||TB)=\rH_{\mathbf{p}}(A||B)$$
if and only if the following conditions hold:
\begin{enumerate}[(i)]
\item $\mathbf{a}_j = \bigoplus_{k=1}^{K}\mu^{(j)}_k \mathbf{p}^{(j)}_k \ot \mathbf{r}_{k}$ and $\mathbf{b}_j = \bigoplus_{k=1}^{K}\nu^{(j)}_k \mathbf{q}^{(j)}_k \ot \mathbf{r}_{k}$,
where $\mathbf{p}^{(j)}_k,\mathbf{q}^{(j)}_k$ denote
$m_{k}$-dimensional probability vectors, and $\mathbf{r}_{k}$ are
$n_{k}$-dimensional probability vectors and $\forall k:
\mu_{k},\nu_{k}\geqslant0, \sum_{k=1}^K \mu^{(j)}_k=\sum_{k=1}^K
\nu^{(j)}_k=1$ for each  $j=1,\ldots,N$, $\sum_{k=1}^K m_k n_k=N$;
\item $T = \bigoplus_{k=1}^{K}\pi_k \ot T_{k}$, where $\pi_{k} \in \mathbf{S}_{m_k}$ and $T_k$ is $n_k\times n_k$ stochastic matrix for each $k$.
\end{enumerate}
\end{thrm}


\begin{proof}
By the definition of relative entropy for stochastic matrices, it
follows that
$$\rH_{p}(A||B) = \sum_{j=1}^{N}p_{j}\rH(\mathbf{a}_{j}||\mathbf{b}_{j}),$$
where $\mathbf{a}_{j} = [a_{1j}, \ldots, a_{Nj}]^{\t}$ and
$\mathbf{b}_{j} = [b_{1j}, \ldots, b_{Nj}]^{\t}$ are $j$th columns
of $A$ and $B$, respectively. Now
\begin{eqnarray*}
\rH_{\mathbf{p}}(TA||TB) & = & \sum_{j=1}^{N}p_{j}\rH((TA)_{j}||(TB)_{j}) = \sum_{j=1}^{N}p_{j}\rH(T\mathbf{a}_{j}||T\mathbf{b}_{j})\\
& \leqslant &
\sum_{j=1}^{N}p_{j}\rH(\mathbf{a}_{j}||\mathbf{b}_{j})=\rH_{\mathbf{p}}(A||B).
\end{eqnarray*}
Thus it follows from the above process in the proof that when the
components of $p$ are all positive,
$\rH_{\mathbf{p}}(TA||TB)=\rH_{\mathbf{p}}(A||B)$ if and only if
$\rH(T\mathbf{a}_j||T\mathbf{b}_j)=\rH(\mathbf{a}_j||\mathbf{b}_j)$
for each $j$. By Theorem \ref{th:Sentropy:equality}, the equality
condition can be concluded immediately.
\end{proof}


\begin{remark}

Now denote $L^{(k)} =
[\mathbf{p}^{(1)}_{k},\ldots,\mathbf{p}^{(N)}_{k}]$ and $R^{(k)} =
[\mathbf{q}^{(1)}_{k},\ldots,\mathbf{q}^{(N)}_{k}]$. Let $E^{(k)} =
Diag[\mu^{(1)}_k,\ldots,\mu^{(N)}_k]$ and $F^{(k)} =
Diag[\nu^{(1)}_k,\ldots,\nu^{(N)}_k]$. The explicit forms of $A,B$
can be written as \\ \indent
$A = \bigoplus_{k=1}^{K} E^{(k)} L^{(k)} \ot \mathbf{r}_{k}$ and $B = \bigoplus_{k=1}^{K} F^{(k)} R^{(k)} \ot \mathbf{r}_{k}$, \\
where $L^{(k)}$ and $R^{(k)}$ are any stochastic matrices, where
$k=1,\ldots,K$. Furthermore, $\sum_{k=1}^K E^{(k)} = \sum_{k=1}^K
F^{(k)} = \I$.

For a finite collection $\set{B^{(i)}}$ of $N\times N$ stochastic
matrices, denote $\overline{B}=\sum_{i}\lambda_{i}B^{(i)}$, where
$\set{\lambda_{i}}$ is a probability vector. The
\emph{$\chi$-quantity} for $\set{B^{(i)}}$ is defined by
$\chi_{\mathbf{p}}(\set{\lambda_{i}, B^{(i)}}) = \sum_{i}
\lambda_{i}\rH_{\mathbf{p}}(B^{(i)}||\overline{B})$, where
$\mathbf{p}$ is a probability vector. It is easily seen from Theorem
\ref{Mon:stochasticmatrices} that
\begin{enumerate}[(i)]
\item $\chi_{\mathbf{p}}(\set{\lambda_{i},B^{(i)}}) = \rH_{\mathbf{p}}(\sum_{i}\lambda_{i}B^{(i)}) - \sum_{i}\lambda_{i}\rH_{\mathbf{p}}(B^{(i)})$;

\item $\sum_{i}\lambda_{i}\rH_{\mathbf{p}}(B^{(i)}||D) = \chi_{\mathbf{p}}(\{\lambda_{i},B^{(i)}\}) + \rH_{\mathbf{p}}(\overline{B}||D)$, i.e.,
$\sum_{i}\lambda_{i}\rH_{\mathbf{p}}(B^{(i)}||D) =
\sum_{i}\lambda_{i}\rH_{\mathbf{p}}(B^{(i)}||\overline{B})+\rH_{\mathbf{p}}(\overline{B}||D)$,
where $D$ is an $N\times N$ stochastic matrix and $p$ is an
$N$-dimensional probability vector.

\item Assume that $T$ is an $N\times N$ stochastic matrix.
Then $\chi_{\mathbf{p}}(\set{\lambda_{i},TB^{(i)}}) \leqslant
\chi_{\mathbf{p}}(\set{\lambda_{i},B^{(i)}})$ if and only if
$\rH_{\mathbf{p}}(TB) - \rH_{\mathbf{p}}(B)$ is a convex function in
its argument stochastic matrix $B$; moreover,
$\chi_{\mathbf{p}}(\set{\lambda_{i},TB^{(i)}}) \leqslant
\chi_{\mathbf{p}}(\set{\lambda_{i},B^{(i)}})$.
\end{enumerate}

\end{remark}


In \cite{Slomczynski}, W. S{\l}omczy\'{n}ski obtained that given any
$N\times N$ stochastic matrices $X, Y, Z$ for which $\mathbf{p}$ is
their common invariant probability vector, i.e. $X\mathbf{p} =
Y\mathbf{p} = Z\mathbf{p} = \mathbf{p}$. Then :
\begin{enumerate}[(i)]
\item $\rH_{\mathbf{p}}(Y) \leqslant \rH_{\mathbf{p}}(XY) \leqslant\rH_{\mathbf{p}}(X) + \rH_{\mathbf{p}}(Y)$;
\item $\rH_{\mathbf{p}}(XYZ) + \rH_{\mathbf{p}}(Y) \leqslant \rH_{\mathbf{p}}(XY) + \rH_{\mathbf{p}}(YZ)$.
\end{enumerate}
The following result is to deal with the saturation of the above two
inequalities.

\begin{prop}
\begin{enumerate}[(i)]
\item If $T\in\mathbf{B}_N$, $A$ is an $N\times N$ stochastic matrix and
$\mathbf{p}$ is an $N$-dimensional probability vector with all
positive components, then $\rH_\mathbf{p}(TA) = \rH_\mathbf{p}(A)$
if and only if $T^\t TA = A$;

\item If $X = X_{L}\ot\pi_{R}$ and $Y = \pi_{L}\ot Y_{R}$
for $X_{L}$ being stochastic matrix of size $m\times m$, $\pi_L \in
\mathbf{S}_m$, $Y_R$ being stochastic matrix of size $n\times n$,
$\pi_R \in \mathbf{S}_n$, then $\rH(XY) = \rH(X) + \rH(Y)$;

\item If $X = \bigoplus_{k=1}^{K} X^L_k \ot \pi^R_k, Y = \bigoplus_{k=1}^{K}Y^L_k \ot Y^R_k$
and $Z = \bigoplus_{k=1}^{K}\pi^L_k \ot Z^R_k$ for $X^L_k$ being
stochastic matrix of size $m_k\times m_k$, $\pi^L_k \in
\mathbf{S}_{m_k}$, $Y^R_k$ being stochastic matrix of size
$n_k\times n_k$, $\pi^R_k \in \mathbf{S}_{n_k}$, then $\rH(XYZ) +
\rH(Y) = \rH(XY) + \rH(YZ)$.
\end{enumerate}
\end{prop}

\begin{proof}
\begin{enumerate}[(i)]
\item Since each component $p_j$ is positive, it follows that
$\rH_\mathbf{p}(TA) = \rH_\mathbf{p}(A)$ if and only if
$\rH(T\mathbf{a}_j) = \rH(\mathbf{a}_j)$ for every $j$, where all
$\mathbf{a}_j$'s are the $j$th column vector of $A$. By the result
in \cite{Alan,Zhang}, i.e. for $B\in\mathbf{B}_N$, $\rH(B\mathbf{p})
= \rH(\mathbf{p})$ if and only if $B^\t B\mathbf{p} = \mathbf{p}$,
we get that $\rH(T\mathbf{a}_j) = \rH(\mathbf{a}_j)$ for every $j$
if and only if $T^\t T\mathbf{a}_j = \mathbf{a}_j$ for all $j$; that
is, the proof is concluded.

\item Since $XY = X_L\pi_L \ot \pi_RY_R$, it follows that
$\rH(XY)  =  \rH(X_L\pi_L \ot \pi_RY_R) = \rH(X_L) + \rH(Y_R)$,
which implies the conclusion.

\item Since
$$XYZ = \bigoplus_{k=1}^{K}X^L_kY^L_k\pi^L_k  \ot \pi^R_kY^R_k Z^R_k,$$
it follows that \begin{eqnarray*}
\rH(XYZ) & = & \sum_k \lambda_k [\rH(X^L_kY^L_k) + \rH(Y^R_k Z^R_k)],\\
\rH(XY) & = & \sum_k \lambda_k [\rH(X^L_kY^L_k) + \rH(Y^R_k)],\\
\rH(Y Z) & = & \sum_k \lambda_k [\rH(Y^L_k) + \rH(Y^R_k Z^R_k)],\\
\rH(Y)  & = &  \sum_k \lambda_k [\rH(Y^L_k) + \rH(Y^R_k)],
\end{eqnarray*}
where $\lambda_k=m_kn_k/N$ and $\sum_k m_kn_k = N$. Combining all
these expressions gives the desired result.
\end{enumerate}
\end{proof}


\section{The relationship between quantum operations and bi-stochastic matrices}


For any CP super-operators $\Phi$ and $\Psi$, with corresponding
Kraus representations: $\Phi = \sum_{i}\mathbf{Ad}_{M_{i}}$ and
$\Psi = \sum_{j}\mathbf{Ad}_{N_{j}}$, respectively, it is easily
seen that $\Phi \ot \Psi = \sum_{i,j}\mathbf{Ad}_{M_{i} \ot N_{j}}$,
and $\jam{\Phi \ot \Psi} = \jam{\Phi}\ot \jam{\Psi}$; denote
$\Psi^\t = \sum_{j}\mathbf{Ad}_{N^{\t}_{j}}$. Then
$$
\jam{\Phi\circ\Psi} = \Phi \ot \I(\jam{\Psi}) =
\I\ot\Psi^\t(\jam{\Phi})
 = \Phi \ot \Psi^\t(\col{\I}\row{\I}).
$$

Let $\Phi,\Psi\in\CP(\cH)$ be stochastic. The \emph{relative
entropy} between $\Phi$ and $\Psi$ is defined by
$$
\rS(\Phi||\Psi)\defeq\rS(\rho(\Phi)||\rho(\Psi))
$$
If
$\Lambda\in\CP(\cH)$ is also bi-stochastic, then
$$
\rS(\Lambda\circ\Phi||\Lambda\circ\Psi) \leqslant \rS(\Phi||\Psi).
$$
This can be seen easily from the Lemma \ref{Lind}. Indeed,
$$
\rS(\Lambda\ot\I_{\lin{\cH}}(\rho(\Phi))||\Lambda\ot\I_{\lin{\cH}}(\rho(\Psi)))
\leqslant \rS(\rho(\Phi)||\rho(\Psi)) = \rS(\Phi||\Psi)
$$
since
$\Lambda\ot\I_{\lin{\cH}}$ is bi-stochastic whenever $\Lambda$ is
bi-stochastic.

Assume that $\Phi$ is a CP stochastic super-operator for which the
Kraus decomposition can be written as
$\Phi=\sum_{\mu}\mathbf{Ad}_{T_{\mu}}$. Define the \emph{Kraus
matrix} \cite{Bengtsson} for $\Phi$ as
$$
B(\Phi):=\sum_{\mu}T_{\mu}\bullet T^{\ast}_{\mu},
$$
where $\bullet$
denotes Shur product of matrices and $\ast$ means that entry-wise
complex conjugate of a matrix. Hence the $(i,j)$th entry $b_{ij}$ of
$B$ can be described by
$b_{ij}=\sum_{\mu}t^{\mu}_{ij}\overline{t^{\mu}_{ij}}$, where
$T_{\mu}=[t^{\mu}_{ij}]$ and the bar means the complex conjugate of
complex numbers.

For any two Hermitian matrices $X$ and $Y$, $X$ is \emph{majorized}
by $Y$, denoted by $X\prec Y$, if there is a CP bi-stochastic
super-operator $\Phi$ such that $X = \Phi(Y)$. The well-known Shur's
theorem states that $Diag(X)\prec X$ for any square matrix $X$, see
\cite{Bengtsson}. Thus for any bi-stochastic quantum operation
$\Lambda$, it follows that $\Lambda(\rho)\prec\rho$. Moreover,
$$
\jam{\Phi\circ\Psi}\prec\jam{\Phi},\quad
\jam{\Phi\circ\Psi}\prec\jam{\Psi}
$$
for any two bi-stochastic
quantum operations $\Phi$ and $\Psi$.

In what follows, some properties of Kraus matrices are listed below.

\begin{prop}\label{extreme point}
\begin{enumerate}[(i)]

\item For a given (bi-)stochastic super-operator $\Phi\in\CP(\cH)$, $B(\Phi)$ is a (bi-)stochastic matrix.

\item $B(\Phi)$ is well-defined, i.e., it is independent of the different Kraus decompositions for $\Phi$ and just depends on $\Phi$ itself.

\item $B(\Phi)$ is a convex function with respect to argument $\Phi$, i.e., \\ \indent $B(t\Phi_1+(1-t)\Phi_2)=tB(\Phi_1)+(1-t)B(\Phi_2)$
for any $\Phi_1,\Phi_2$ and all $t\in[0,1]$.

\item Denote $\sM(B)=\set{\Phi| \Phi\in\trans{\cH} \mbox{\ is CP stochastic super-operator and\ } B(\Phi) = B}$. Then $\sM(B)$ is a nonempty convex set.

\item $B(\Theta_1\ot\Theta_2) = B(\Theta_1)\ot B(\Theta_2)$ for any stochastic quantum operations $\Theta_1$ and $\Theta_2$.

\item Assume that $\cH$ and $\cK$ are $M$ and $N$ dimensional Hilbert spaces, respectively.
If $\Lambda \in \trans{\cH \ot \cK}$ is CP and stochastic  and can
be described by $\Lambda = \sum_k \lambda_k \Phi_k \ot \Psi_k$,
where $\set{\Phi_k} \in \trans{\cH}$ and $\set{\Psi_k} \in
\trans{\cK}$ are two collections of CP stochastic super-operators,
then $B(\Lambda) = \sum_k \lambda_k B(\Phi_k) \ot B(\Psi_k)$, where
$\lambda=\set{\lambda_k}_k$ is a finite probability vector.
\end{enumerate}
\end{prop}


\begin{proof}
\begin{enumerate}[(i)]

\item The proof is trivial.

\item Assume that $\Phi = \sum_{\mu=1}^{N^2}\mathbf{Ad}_{E_{\mu}} = \sum_{\nu=1}^{N^2}\mathbf{Ad}_{F_{\nu}}$. By the unitary freedom of quantum operations, there is a $N^2\times N^2$ unitary matrix $U=[u_{\mu\nu}]$ such that $E_{\mu} = \sum_{\nu=1}^{N^2}u_{\mu\nu}F_{\nu}$. Then
\begin{eqnarray*}
\sum_{\mu=1}^{N^2}E_{\mu} \bullet E^{\ast}_{\mu} & = &
\sum_{\mu=1}^{N^2}(\sum_{\nu=1}^{N^2}u_{\mu\nu}F_{\nu}) \bullet
(\sum_{\kappa=1}^{N^2}u_{\mu\kappa}F_{\kappa})^{\ast} =
\sum_{\nu,\kappa=1}^{N^2}(\sum_{\mu=1}^{N^2}u_{\mu\nu}\overline{u}_{\mu\kappa})F_{\nu}\bullet
F^{\ast}_{\kappa}\\ & = &
\sum_{\nu,\kappa=1}^{N^2}\delta_{\nu\kappa}F_{\nu}\bullet
F^{\ast}_{\kappa} = \sum_{\nu=1}^{N^2}F_{\nu}\bullet F^{\ast}_{\nu},
\end{eqnarray*}
which implies that $B(\Phi)$ is well-defined.

\item Choose any two stochastic quantum operations $\Phi_1$ and $\Phi_2$
with their corresponding Kraus decomposition: $\Phi_1 =
\sum_{\mu}\mathbf{Ad}_{S_{\mu}}$ and $\Phi_2 =
\sum_{\nu}\mathbf{Ad}_{T_{\nu}}$. Let $t\in[0,1]$. Then the Kraus
decomposition for $t\Phi_1 + (1-t)\Phi_2$ is $t\Phi_1 + (1-t)\Phi_2
= t\sum_{\mu}\mathbf{Ad}_{S_{\mu}} +
(1-t)\sum_{\nu}\mathbf{Ad}_{T_{\nu}}$, which implies that the Kraus
matrix for $t\Phi_1 + (1-t)\Phi_2$ is
\begin{eqnarray*}
B(t\Phi_1 + (1-t)\Phi_2) & = & \sum_{\mu}(\sqrt{t}S_{\mu}) \bullet (\sqrt{t}S^{\ast}_{\mu}) + \sum_{\nu}(\sqrt{1-t}T_{\nu}) \bullet (\sqrt{1-t}T^{\ast}_{\nu})\\
& = & t\sum_{\mu}S_{\mu}\bullet S^{\ast}_{\mu} +
(1-t)\sum_{\nu}T_{\nu} \bullet T^{\ast}_{\nu}\\ & = & tB(\Phi_1) +
(1-t)B(\Phi_2).
\end{eqnarray*}

\item If $\Psi_1,\Psi_2 \in \sM(B)$, it follows from the result of (iii) that $B(t\Psi_1+(1-t)\Psi_2)=tB(\Psi_1)+(1-t)B(\Psi_2)=B$ since $B(\Psi_1)=B(\Psi_2)=B$, which implies $t\Psi_1+(1-t)\Psi_2\in\sM(B)$. The fact that $\sM(B)$ is not empty is clear.

\item Let the Kraus decompositions for $\Theta_1$ and $\Theta_2$ are $\Theta_1=\sum_{m}\mathbf{Ad}_{S_{m}}$ and
$\Theta_2=\sum_{\mu}\mathbf{Ad}_{T_{\mu}}$. Then
$\Theta_1\ot\Theta_2=\sum_{m,\mu}Ad_{S_{m}\ot T_{\mu}}$. Now
\begin{eqnarray*}
B(\Theta_1\ot\Theta_2) & = & \sum_{m,\mu}(S_{m}\ot T_{\mu})\bullet(S^{\ast}_{m}\ot T^{\ast}_{\mu}) = (\sum_{m}S_{m}\bullet S^{\ast}_{m}) \ot (\sum_{\mu}T_{\mu}\bullet T^{\ast}_{\mu})\\
& = & B(\Theta_1) \ot B(\Theta_2).
\end{eqnarray*}

\item It follows trivially from combining the above conclusions (iii) and (v).
\end{enumerate}
\end{proof}


\begin{remark}
In the above Proposition \ref{extreme point}(iv), it is known that
$\sM(B)$ is a nonempty convex set. In fact, it is also compact. Thus
the question naturally arises: what are the extreme points of
$\sM(B)$? Note that in \cite{Parthasarathy}, Parthasarathy gave a
characterization of extremal quantum states of composite systems
with fixed marginal states. Subsequently, Rudolph gave an another
characterization about it in \cite{Rudolph}. Therefore, our question
can be described in terms of the language as in
\cite{Parthasarathy,Rudolph} under the additional condition that the
diagonals of the Jamio{\l}kowski state are fixed.
\end{remark}


\begin{remark}

Generally speaking, $B(\Phi\circ\Psi)\neq B(\Phi)B(\Psi)$ for two
stochastic super-operators $\Phi,\Psi \in \CP(\cH)$. But there is a
class of specific examples in which the inequality sign can be
replaced by an equal sign. Indeed, assume that $\cH$ is an
$N$-dimensional Hilbert space and $\Phi, \Psi \in \CP(\cH)$ for
which
$$
\jam{\Phi} = \sum_{m,\mu=1}^N p_{m\mu}\out{m\mu}{m\mu}\quad\text{and}\quad\jam{\Psi} = \sum_{m\mu=1}^N q_{m,\mu}\out{m\mu}{m\mu}.
$$
By the stochasticity, $\sum_{\mu=1}^{N}p_{m\mu}=1$ and
$\sum_{m=1}^{N}p_{m\mu}=1$; $\sum_{\mu=1}^{N}q_{m\mu}=1$ and
$\sum_{m=1}^{N}q_{m\mu}=1$. Then $\jam{\Phi\circ\Psi} =
\sum_{m,\mu=1}^{N}[B(\Phi)B(\Psi)]_{m\mu}\out{m\mu}{m\mu}$, where
$B(\Phi) = [p_{m\mu}]$ and $B(\Psi) = [q_{m\mu}]$, which implies
that
$$
B(\Phi\circ\Psi) = B(\Phi)B(\Psi),\quad B(\Psi\circ\Phi) =
B(\Psi)B(\Phi)
$$
and
$$
\rS(\Phi\circ\Psi) = \rH(B(\Phi\circ\Psi)) + \log N,\quad
\rS(\Psi\circ\Phi) = \rH(B(\Psi\circ\Phi)) + \log N.
$$
Now
$$
\rS^{\mathbf{map}}(\Phi) + \rS^{\mathbf{map}}(\Psi) -
\rS^{\mathbf{map}}(\Phi\circ\Psi) = \rH(B(\Phi)) + \rH(B(\Psi)) -
\rH(B(\Psi)B(\Phi)) + \log N.
$$

This fact shows
that if both $J(\Phi)$ and $J(\Psi)$ are diagonal, then
$B(\Phi\circ\Psi) = B(\Phi)B(\Psi)$. There is a question which can
be formulated as follows: what is a sufficient and necessary
condition for $B(\Phi\circ\Psi) = B(\Phi)B(\Psi)$ for stochastic
super-operators $\Phi,\Psi\in\CP(\cH)$. It is also \emph{conjectured}
that $\rH_{\mathbf{p}}(B(\Phi\circ\Psi)) \leqslant
\rH_{\mathbf{p}}(B(\Phi)B(\Psi))$ for any stochastic super-operators
$\Phi,\Psi\in\CP(\cH)$, where $\mathbf{p}$ is any $N$-dimensional
probability vector.
\end{remark}


\begin{prop} Assume that $\cH$ is a $N$-dimensional Hilbert space.
\begin{enumerate}[(i)]
\item If $\Phi \in \CP(\cH)$ is stochastic, then:
$\rS^{\mathbf{map}}(\Phi) \leqslant \rH(B(\Phi))+\log N$;
\item If $\Phi,\Psi \in \CP(\cH)$ is stochastic, then: $\rH(B(\Phi)||B(\Psi)) \leqslant \rS(\Phi||\Psi)$;
\end{enumerate}
\end{prop}


\begin{proof}
\begin{enumerate}[(i)]
\item By Shur's lemma, it follows that $Diag(J(\Phi))\prec J(\Phi)$ which is equivalent to $Diag(\rho(\Phi))\prec \rho(\Phi)$.
Since $\Innerm{m}{B(\Phi)}{\mu} = \Innerm{m\mu}{J(\Phi)}{m\mu}$, it
can be seen that
$$
\rS^{\mathbf{map}}(\Phi) = \rS(\rho(\Phi)) \leqslant
\rS(Diag(\rho(\Phi))) = \rH(B(\Phi)) + \log N.
$$
Furthermore, $\rS^{\mathbf{map}}(\Phi) = \rH(B(\Phi)) + \log N$ when
$\Phi$ is represented by a diagonal dynamical matrix $J(\Phi)$.

\item There exists a CP bi-stochastic super-operator $\Lambda$ such that $\Lambda(\rho)=Diag(\rho)$ since
$Diag(\rho)\prec\rho$ which follows from Shur's lemma. Thus it follows from Lemma \ref{Lind} that
\begin{eqnarray*}
\rS(\Phi||\Psi)&=&\rS(\rho(\Phi)||\rho(\Psi))\geqslant\rS(Diag[\rho(\Phi)]||Diag[\rho(\Psi)])\\
&=&\frac1N\sum_{j=1}^{N}\rH(B(\Phi)_j||B(\Psi)_j)=\rH(B(\Phi)||B(\Psi)).
\end{eqnarray*}
\end{enumerate}
\end{proof}

\begin{remark}
Recently, Roga \emph{et al.} studied entropic uncertainty relations
for quantum operations in \cite{Roga2012} and related bounds on the
map entropy were also obtained.

Let $\widetilde{s}(\Phi) = \rH(B(\Phi)) - \rS^{\mathbf{map}}(\Phi)$
for stochastic super-operator $\Phi \in \CP(\cH)$. Then for a
collection $\set{\Phi_k}$ of stochastic super-operator in $\CP(\cH)$
such that $\Phi = \sum_k \lambda_k\Phi_k$,
$\chi(\set{\lambda_{k},B(\Phi_{k})})\leqslant\chi(\set{\lambda_{k},\Phi_{k}})$
if and only if
$\widetilde{s}(\sum_{k}\lambda_{k}\Phi_{k})\leqslant\sum_{k}\lambda_{k}\widetilde{s}(\Phi_{k})$;
i.e., $\widetilde{s}(\Phi)$ is a convex function in its argument
$\Phi$.
\end{remark}

\subsection*{Acknowledgement}

The author would like to thank the anonymous referee(s) for valuable
comments and to thank KMR Audenaert for pointing out some misprints
appearing in this paper. LZ is also grateful for funding from
Hangzhou Dianzi University (Grant No. KYS075612038).


\end{document}